\documentclass[12pt]{article}

\usepackage[margin=1in]{geometry}
\usepackage{amsmath,amsfonts,graphicx,framed}
\usepackage{amssymb}
\usepackage{amsthm}
\usepackage{fancyhdr}
\usepackage{float}
\usepackage{caption}
\usepackage{comment}
\usepackage[utf8]{inputenc}
\usepackage{authblk}

\usepackage[hypertexnames=false]{hyperref}

\theoremstyle{plain}
\newtheorem{thm}{Theorem}

\newtheorem{prop}{Proposition}
\theoremstyle{definition}
\newtheorem{defn}{Definition}

\theoremstyle{remark}
\newtheorem{rem}{Remark}

\numberwithin{equation}{section}

\newcommand{\X}{\,\text{X}}
\newcommand{\Xe}{\,\emph{X}}
\newcommand{\pd}{\,\partial}
\newcommand{\ov}{\overline}
\newcommand{\di}{\,\mathrm d}
\newcommand{\mbb}{\mathbb}
\newcommand{\pdf}[2][]{\,\frac{\partial #1}{\partial #2}}
\newcommand{\eps}{\,\varepsilon}

\newcommand{\al}{\,\alpha}
\newcommand{\mc}{\mathcal}
\newcommand{\De}{\,\Delta}
\newcommand{\Y}{\,\text{Y}}
\newcommand{\eqal}[1]{\begin{equation}\begin{aligned}#1\end{aligned}\end{equation}}

\begin{document}

\title{Sub-symmetries I. Main properties and applications.}

\author{V. Rosenhaus*}
\author{Ravi Shankar}

\affil{\footnotesize{* Department of Mathematics and Statistics, California State University, Chico, CA, USA, vrosenhaus@csuchico.edu}}
\affil{\footnotesize{Department of Mathematics, University of Washington, Seattle, WA, USA, 
shankarr@uw.edu}}

\date{}

\maketitle

\begin{abstract}
We introduce a \textit{sub-symmetry} of a differential system as an infinitesimal transformation of a subset of the system that leaves the subset invariant on the solution set of the entire system. We discuss the geometrical meaning and properties of sub-symmetries, as well as an algorithm for finding sub-symmetries of a system. We show some of the benefits of using sub-symmetries in the search for solutions of a system; 
in particular, we show how sub-symmetries can be used in decoupling a differential system. We also discuss the role of sub-symmetries in the deformation of known conservation laws of a system into other (often, new) conservation laws and show that, in this regard, a sub-symmetry is a considerably more powerful tool than a regular symmetry. We demonstrate that all lower conservation laws of the nonlinear telegraph system can be generated by sub-symmetries. 
\end{abstract}

\section{Introduction}

The Lie Theory of group properties of differential equations is known to provide universal and powerful methods and techniques to study many properties of differential systems: analytical solutions and reduced systems, transformations between solutions of a given PDE, mappings to different equations, the linearization problem, conservation laws, etc., see e.g. \cite{Olver}, or \cite{Bluman_Chev}. The motivation of our work is to show that, in some aspects, sub-symmetries of differential systems provide a more effective foundation and tools than regular symmetries.

\smallskip

In this paper, we demonstrate two of those aspects: decoupling of a differential system, and deformation of its conservation laws.
Our main and original motivation is in the natural connection that sub-symmetries provide with local conservation laws of non-variational differential systems through the Noether identity \cite{Rosen} and a Noether-type theorem; this problem will be discussed in our next paper on sub-symmetries, \cite{RSII2016}. 

\medskip

In the paper, we introduce the concept of a sub-symmetry of a differential system. A regular symmetry is a transformation of a differential system that leaves all equations of the system invariant on solutions of the system, see e.g. \cite{Olver} or \cite{Bluman_Chev}. Here we consider transformations that characterize invariance properties of a part of the system (a sub-system). By a \textit{sub-symmetry}, we mean transformations that leave a \textit{part} of the system invariant on solutions of the system. Thus, if symmetry transformations leave all equations of the system invariant (on solutions of the system), sub-symmetries leave only some equations of the system or their combinations invariant on solutions of the entire system. 

\medskip

In Section \ref{sec:subsym}, we start with a qualitative description of sub-symmetries for a system of two differential equations. We then discuss the geometrical meaning of sub-symmetry transformations, different types of sub-symmetries, and their examples. We give a definition of  sub-symmetry transformations, discuss some of their properties, and examine relations between sub-symmetries, symmetries,  and conditional symmetries. We also discuss an algorithm for finding sub-symmetries, and demonstrate sets of sub-symmetries for some known differential systems.

In Section \ref{sec:decoup}, we study the possibility of decoupling of a sub-system within the system and the role sub-symmetries play in this decoupling. We illustrate our results with some examples of the use of sub-symmetries to decouple differential systems.  Our discussion is restricted to PDEs involving $2+2$ independent and dependent variables, but the results can be naturally extended to arbitrary dimension.

In Section \ref{sec:deform}, 
we consider the deformation of known conservation laws into new conservation laws using infinitesimal transformations.  
We show that sub-symmetries could play a considerably more powerful role in this regard.  
We demonstrate that for the nonlinear telegraph system, all lower conservation laws can be obtained by sub-symmetry deformations.
 
\section{Sub-symmetry. Definition and main properties.}
\label{sec:subsym}
\subsection{Overview of invariance conditions}
Consider an infinitesimal transformation with operator $\text{X}$ 
\begin{equation}\label{non-canon_X}
\text{X}=\xi^i\frac{\partial}{\partial x_i}+\eta^a\frac{\partial}{\partial{u^a}}+ \zeta^a_i\frac{\partial}{\partial{u^a_i}}+\ldots, 
\end{equation}
where  $x=(x_1,\dots,x_p)$, $u=(u^1,\dots,u^q)$, $\text{X}$ is a correspondingly prolonged vector field, see e.g. \cite{Olver}, and repeated indices imply summation. Without loss of generality we can consider instead a corresponding canonical (evolutionary) operator  
\begin{align}
\X_{\alpha}=\text{X}-\xi^iD_i,
\end{align}
where  $\alpha^a=\eta^a-\xi^iu_i^a$.
In what follows,
unless otherwise specified, by an infinitesimal operator  $\text{X}$ we mean a vertical vector field $\X_{\alpha}\; (\X)$ acting on the space of dependent variables $u^a, \: a=1,\dots,n$
\begin{align}\label{inf_op}
\X=\al^a\frac{\partial}{\partial{u^a}}+\left(D_j\al^a\right)\frac{\partial}{\partial{u^a_j}}+\dots,
\end{align}
where $\al^a=\al^a(x,u^b,u^b_j,\dots)$ is the infinitesimal, and $D_j$ is a total derivative:
\begin{align}
D_j=\pdf{x_j}+u^a_j\pdf{u^a}+\dots.
\end{align}
Geometrically, vector fields can be thought of as infinitesimal transformations for a family of finite, continuous transformations:
\begin{align}\label{fin}
e^{\eps \X}=1+\eps\X+\frac{1}{2}\eps^2\X^2+\dots,\qquad \eps\in (-\infty,\infty).
\end{align}
We identify this family of transformations with its infinitesimal generator $\X$. 

\medskip

Consider a system $\Delta(x,u,\dots)\equiv \Delta[u]$ of two partial differential equations for two functions $u^1,u^2$ of two variables $(x_1,x_2)$:
\begin{align}\label{2sys}
\begin{split}
\Delta_1(x_i,u^a,u^a_j,\dots)&=0,\\
\Delta_2(x_i,u^a,u^a_j,\dots)&=0.
\end{split}
\end{align}
If transformation $\X$ is a \textit{symmetry} of the system \eqref{2sys} then the following invariance conditions are satisfied:
\begin{align}\label{2sym}
\X\Delta_i\bigl|_{\Delta_1=\Delta_2=0}=0,\qquad i=1,2.
\end{align}
Since $\X D_i = D_i\X$, (see e.g. \cite{Olver}) then
\begin{align}\label{groupp}
\Delta[e^{\eps\X}\,u]=e^{\eps\X}\Delta[u]=(1+\eps\X+\frac{1}{2}\eps^2\X^2+\dots)\Delta[u].
\end{align}
According to invariance conditions \eqref{2sym}, $e^{\eps\X}u$ is a solution to \eqref{2sys} whenever $u$ is a solution, and  $e^{\eps\X}\Delta_i=0$ when $\Delta_i=0$.  Geometrically, these conditions reflect a known fact that finite continuous symmetry transformations $e^{\eps\X}$ map solutions $u$ of the system $\Delta[u]$ to its other (in general) solutions, and therefore, leave the solution set of \eqref{2sys} invariant.

It was shown in \cite{Olver} that symmetry condition \eqref{2sym} can be written more explicitly in terms of the $\Delta$'s themselves (provided our system is nondegenerate, which we assume, meaning that the system is locally solvable at every point, and of maximal rank, see \cite{Olver} ). Specifically, we can find some functions $\Gamma^{ijI}(x,u,u_{(1)},\dots)$ such that
\begin{align}\label{2exp}
\X\Delta_i = \Gamma^{ij0}\Delta_j+\Gamma^{ijk}D_k\Delta_j+\dots,\qquad i=1,2.
\end{align}
Since the right hand side vanishes on solutions $u$, equality \eqref{2exp} implies symmetry condition \eqref{2sym}.  Note that such an equality holds for each equation of the system $\Delta_i$, and in general, all $\Delta_j$'s can be present on the right hand side.

Let us extend our set of invariant transformations.  Let us require an invariance of only a \textit{subsystem} of \eqref{2sys} under these transformations, e.g., $\Delta_1$. Consider vector fields $\X$ that satisfy
\begin{align}\label{2sub}
\X\Delta_1\bigl|_{\Delta_1=\Delta_2=0}=0.
\end{align}
Instead of conditions \eqref{2exp} infinitesimal transformations $\X$ now need to satisfy only 
\begin{align}\label{2sube}
\X\Delta_1 = \Gamma^{1j0}\Delta_j+\Gamma^{1jk}D_k\Delta_j+\dots,
\end{align}
(with some functions $\Gamma^{1jI}(x,u,u_{(1)},\dots)$) in order to insure invariance of the first equation of the system \eqref{2sys}. If the invariance condition of a sub-system (e.g. \eqref{2sube}) is satisfied, we say that $\X$ is a \textit{sub-symmetry} of the original system.  

Observe that any symmetry $\X$ of the system \eqref{2sys} satisfies condition \eqref{2sub}, see \eqref{2sym}. Therefore, a symmetry of the system is also its sub-symmetry; the converse statement does not necessarily hold.

To illustrate the difference between sub-symmetries and symmetries, let us consider a special case of invariance condition \eqref{2sube} (with no derivatives), and an example of the action of a sub-symmetry $\X$ on $\Delta_2$:
\begin{align}
\label{inv1}
\X\Delta_1&=a\,\Delta_1+b\,\Delta_2,\\
\label{inv2}
\X\Delta_2&=c\,\Delta_1+d\,\Delta_2+R,
\end{align}
with some function $R$. We define the following solution sets of functions that satisfy our equations:
\begin{align}
\begin{split}
\mc S_1&=\{u\,\bigl|\,\Delta_1[u]=0\},\\
\mc S_2&=\{u\,\bigl|\,\Delta_2[u]=0\}.
\end{split}
\end{align}
Note that $\mc S_1\cap\mc S_2$ is the solution set to our system \eqref{2sys}.  Note also that our consideration covers the case when $a,b,c,$ and $d$ are differential operators.

\bigskip
There are three cases to consider.

\medskip
\noindent
\textbf{Case 1:} $R=0$.  Symmetries.

In this case, invariance conditions \eqref{inv1}-\eqref{inv2} imply \eqref{2sym}, which reflects the fact that $\X$ is tangent to $\mc S_1\cap \mc S_2$ on points of $\mc S_1\cap\mc S_2$. 
A transformation with an infinitesimal operator $\X$ is a symmetry of the system \eqref{2sys}, and the corresponding finite transformation $e^{\eps \X}$ maps solutions $u$ in $\mc S_1\cap \mc S_2$ to solutions in the same set $\mc S_1\cap \mc S_2$.  

\medskip
\noindent
\textbf{Case 2:} $R\neq 0$ and $b=0$.  Sub-system symmetries.

Here, by \eqref{inv2}, symmetry conditions \eqref{2sym} (and \eqref{2exp}) are violated, and a finite transformation $e^{\eps \X}$ does \textit{not} map solutions in $\mc S_1\cap \mc S_2$ back into $\mc S_1\cap \mc S_2$.  In this case, $\X$ is tangent to $\mc S_1$ on points of $\mc S_1$, but not to $\mc S_1\cap\mc S_2$ on $\mc S_1\cap\mc S_2$.  For a solution $u$ of system \eqref{2sys} in $\mc S_1\cap\mc S_2$ ($\Delta_1[u]=\Delta_2[u]=0)$, group property \eqref{groupp} and condition \eqref{inv2} give
\begin{align}\label{away2}
\begin{split}
\Delta_2[e^{\eps\X}\,u]&=(1+\eps\X+O(\eps^2))\Delta_2[u]\\
&=\eps(c\Delta_1[u]+d\Delta_2[u]+R)+O(\eps^2)\\
&=\eps R+O(\eps^2).
\end{split}
\end{align}
Since $R\neq 0$, we see that transformation 
$e^{\eps\X}$ maps $u$ ``away" from $\mc S_2$.  However, condition \eqref{inv1} in this case 
\begin{align}\label{1sym}
\X\Delta_1=a\,\Delta_1,
\end{align}
shows that $\X$ is a \textit{symmetry of a sub-system}, $\Delta_1$.  Transformation $e^{\eps\X}$ maps solutions $u$ from the set $\mc S_1\cap S_2$ into $\mc S_1$, and away from $\mc S_2$.  

\medskip
\noindent
\textbf{Case 3:} $R\neq 0$ and $b\neq 0$.  Other sub-symmetries.

This is the case of minimal invariance, since $\X$ is neither a symmetry of system \eqref{2sys}, nor a symmetry of any sub-system.  Here, $\X$ is only tangent to $\mc S_1$ on points of $\mc S_1\cap\mc S_2$.  Using the group property \eqref{groupp}, the fact that for a solution of the system \eqref{2sys}: $\Delta_1[u]=\Delta_2[u]=0$, and conditions \eqref{inv1}-\eqref{inv2} for mapping of solution $u$ from $\mc S_1\cap \mc S_2$ we obtain
\begin{align}\label{away1}
\begin{split}
\Delta_1[e^{\eps\X}\, u]&=(1+\eps\X+\frac{1}{2}\eps^2\X^2+O(\eps^3))\Delta_1[u]\\
&=\eps(a\Delta_1[u]+b\Delta_2[u])+\frac{1}{2}\eps^2\X(a\Delta_1[u]+b\Delta_2[u])+O(\eps^3)\\
&=\frac{1}{2}\eps^2\left[\X a\,\Delta_1+a(a\Delta_1+b\Delta_2)+\X b\,\Delta_2+\,b\,(c\Delta_1+d\Delta_2+R)\right]+O(\eps^3)\\
&=\frac{1}{2}\eps^2\,b\,R+O(\eps^3).
\end{split}
\end{align}
Since both $b$ and $R$ are nonzero, we see that solution $u$ from $\mc S_1\cap \mc S_2$ is mapped away from both $\mc S_1$ and $\mc S_2$.  However, from \eqref{away1} and \eqref{away2} we conclude that, for small $\eps$, solution $u$ is mapped away from $\mc S_1$ to a lesser extent than from $\mc S_2$, and $e^{\eps\X}\,u$ stays very close to $\mc S_1$.  This ``approximate invariance" is a property particular to sub-symmetries; for a general vector field $\X$, $e^{\eps\X}\, u$ would solve both $\Delta_1$ and $\Delta_2$ only up to $O(\eps)$.  

\bigskip
As a simple example that illustrates each case, let us consider the following system for $u(x,y,z)$:
\begin{align}\label{systriv}
\begin{split}
\Delta_1&=u_x=0,\\
\Delta_2&=u_y=0.
\end{split}
\end{align}
The solution sets are $\mc S_1=\{f(y,z)\}$ (set of functions that depend on $(y,z)$), $\mc S_2=\{g(x,z)\}$, and $\mc S_1\cap\mc S_2=\{h(z)\}$.

Case 1.  A symmetry of system \eqref{systriv} is 
\begin{align}
\X=z\pdf{u}+\pdf{u_z},
\end{align}
which corresponds to $z$ translations in function space: $u\to u+\eps\,z$.  Clearly, $\{f(y,z)+\eps\,z\}=\{f(y,z)\}$, and similarly for $\mc S_2$ and $\mc S_1\cap\mc S_2$.

Case 2.  A symmetry of sub-system $\Delta_1$ is 
\begin{align}
\X=y\pdf{u}+\pdf{u_y},
\end{align}
which corresponds to $y$ translations: $u\to u+\eps\,y$.  Indeed, $\X\Delta_2=0,$ $\X\Delta_2=1$, so this transformation preserves $\mc S_1,$ but does not preserve $\mc S_2$: $\{g(x,z)+\eps\,y\}\neq\{g(x,z)\}$ for $\eps\ne 0$.

Case 3.  An example of a sub-symmetry of $\Delta_1$ that is not a symmetry of any sub-system of \eqref{systriv} is
\begin{align}\label{subtriv}
\X=(y+xu_y)\pdf{u}+(u_y+xu_{xy})\pdf{u_x}+(1+xu_{yy})\pdf{u_y}+\dots.
\end{align}
Indeed, we see that
\begin{align}
\begin{split}
\X\De_1&=xD_y\De_1+\De_2,\\
\X\De_2&=xD_y\De_2+1,
\end{split}
\end{align}
such that $a=xD_y,b=1,c=0,d=xD_y,$ and $R=1$ in \eqref{inv1}-\eqref{inv2}.  Thus, $e^{\eps\X}$ maps $u=h(z)$ away from $\mc S_1\cap\mc S_2$. This conclusion can be demonstrated explicitly.  

Let us find the non-canonical operator (horizontal vector field $\X_h$) corresponding to canonical operator \eqref{subtriv} (vertical vector field $\text X$) 
\begin{align}
\X_h=X-xD_y=-x\pdf{y}+y\pdf{u}+u_y\pdf{u_x} +\pdf{u_y}+\ldots.
\end{align}
The respective finite transformation (the group action on a point $(x,y,z,u)$): $$(x_*(\eps),y_*(\eps),z_*(\eps),u_*(\eps))=e^{\eps\X}(x,y,z,u)$$ is:
\begin{align}
\begin{split}
x_*(\eps)&=x,\\
y_*(\eps)&=y-\eps\,x,\\
z_*(\eps)&=z,\\
u_*(\eps)&=u+\eps\,y-\frac{1}{2}\,\eps^2\,x.
\end{split}
\end{align}
To return to a vertical formulation, we set $u=u(x,y,z)$, express $u_*$ as a function of $(x_*,y_*,z_*)$, and remove *'s. We obtain:

\begin{align}
e^{\eps\X}u(x,y,z)\,=u(x,y+\eps\,x,z)+\eps\,y+\frac{1}{2}\,\eps^2\,x.
\end{align}
For a solution to \eqref{systriv} $u=h(z)$ (with some function $h$) we have
\begin{align}\label{maptriv}
e^{\eps\X}\,h(z)=h(z)+\eps\,y+\frac{1}{2}\,\eps^2\,x.
\end{align}
Then, 
\begin{align}
\Delta_1[e^{\eps\X}h]&=\frac{1}{2}\,\eps^2,\\ 
\Delta_2[e^{\eps\X}h]&=\eps,
\end{align}
which shows that for small $\epsilon,$ $e^{\eps\X}\,u$ solves the first equation $\Delta_1$ ``more accurately" than it does the second equation, $\Delta_2$.

Note that the operator $(x+y)\pdf{u}$ does not satisfy any of the three cases above, and is not a sub-symmetry. The finite transformation $e^{\eps\X}\,u=u+\eps(x+y)$ maps solution $u=h(z)$ away from $\mc S_1$ in a similar manner as from $\mc S_2$.

\subsection{Definition refinement}
Consider now more general sub-systems of system $(\Delta_1,\Delta_2)$ than just $\Delta_1$.  For example, the linear hyperbolic system 
\eqal{
\Delta_1&=u_t+c(x)^2v_x,\\
\Delta_2&=v_t+u_x,
\label{lhs}
}
has a sub-system which can be written terms of $u$ alone:
\eqal{
D_t\Delta_1-c(x)^2D_x\Delta_2=u_{tt}-c(x)^2u_{xx}. 
}
This decoupled sub-system of the linear hyperbolic system \eqref{lhs} could be helpful to consider for generating solutions of the system.

Consider nonlinear Schr\"{o}dinger  (NLS) equations 
\begin{align} \label{NLSC}
\begin{split}
i\psi_{t} + \psi_{xx}  -k \psi^2 \psi^{*}  =  0,  \\
-i\psi^{*}_{t} + \psi^{*}_{xx}  -k \psi {\psi^{*}}^2  =  0,
\end{split}
\end{align}

\noindent where $\psi=\psi(x,t)$ is a
complex-valued function, $\psi^{*}=\psi^{*}(x,t)$ is its conjugate, and 
$k$ is a  real constant.   
\noindent In terms of real-valued functions $u=u(x,t)$, $ v=v(x,t)$: $
\psi= u+iv, \:\: \psi ^{*} =u-iv $
\noindent the nonlinear Schr\"{o}dinger equations take a form
\begin{align}\label{NLSR}
\begin{split}
\Delta_1\,& = \,-v_{t} + u_{xx}  -k u (u^2 +v^2)  =  0, \\
\Delta_2\,& = \,\,\, u_{t} + v_{xx} -k v (u^2 +v^2) =  0.
\end{split}
\end{align}
For the Schr\"odinger system \eqref{NLSR}
it would make sense to consider the sub-system
\eqal{
v\Delta_1-u\Delta_2=-D_t\Big(\frac{1}{2}(u^2+v^2)\Big)+D_x(vu_x-uv_x)
}
since it is a conservation law of the Schr\"odinger system (conservation of the number of particles (mass)). 


In this paper, by a ``sub-system", we mean a linear combination(s) of equations $\beta^i\Delta_i=\beta^1\Delta_1+\beta^2\Delta_2$ with some multipliers $\beta^i[u]$.  To ensure non-degeneracy, we require a maximal rank condition of our sub-system, see \cite{Olver}.

Invariance condition \eqref{2sub} for this sub-system generalizes to
\begin{align}\label{2sub1}
\X(\beta^i\Delta_i)\bigl|_{\Delta_1=\Delta_2=0}\:\:=\:0.
\end{align}
We say the pair $(\X,\beta^i\Delta_i)$ is a \textit{sub-symmetry} of system $(\Delta_1,\Delta_2)$.  Here $\X$ is a sub-symmetry vector field of the sub-system $\beta^i\Delta_i$; it is assumed the sub-system is of maximal rank.  
Note that a sub-symmetry is determined by the form of the sub-system $\beta^i\Delta_i$, and different sub-systems may have different sets of admissible vector fields.

As written, our notion of sub-system depends on how we express the solution set of the system $(\Delta_1,\Delta_2)$ in coordinates.  For example, suppose that $(\X,\beta^i\Delta_i)$ is a sub-symmetry of $(\Delta_1,\Delta_2)$.  If we rewrite the system as $\ov\Delta_i=A^{ij}\Delta_j$ for some smoothly invertible $2\times2$ matrix $A[u]$, then $(\X,\ov\beta\,^i\,\ov\Delta_i)$ would be a sub-symmetry of $(\ov\Delta_1,\ov\Delta_2)$, where the multipliers $\ov\beta\,^i=\beta^j(A^{-1})^{ji}$ transform oppositely to the equations.  

\smallskip
To give an alternative geometric definition, let $\mc S$ be the solution manifold of the system $(\Delta_1,\Delta_2)$, and $\mc R$ be that of the sub-system $\beta^i\Delta_i$.  
Then $\mc S\subset\mc R$ ($\mc S$ is a smoothly embedded submanifold of $\mc R$).  We say that the pair $(\X,\mc R)$ is a \textit{sub-symmetry} of manifold $\mc S$ if vector field $\X$ is tangent to the manifold $\mc R\supset\mc S$ on points of $\mc S$.  

Note that multipliers $\beta^i$'s in definition \eqref{2sub1} could be, in general, differential operators: $\beta^{i}=\beta^{i}_0+\beta^{ij}D_j+\dots$. This more general notion of sub-system arises naturally from the latter definition.  If $\mc S$ is the solution set of the system $(\Delta_1,\Delta_2)$, and $\mc R\supset \mc S$ is that of, say, $\Lambda[u]=0$, then $\Lambda$ vanishes on $\mc S$, i.e. $\Lambda|_{\Delta_1=\Delta_2=0}=0$, since $\mc S\subset\mc R$, which means that
\eqal{
\Lambda[u]=\beta^i_0[u]\Delta_i[u]+\beta^{ij}[u]D_j\Delta_i[u]+\dots
}
for some finite set of multipliers $\beta^i_0,\beta^{ij},\dots$.  Since $\Lambda[u]$ is to be interpreted as a sub-system of $(\Delta_1,\Delta_2)$, it makes sense to consider multipliers as differential operators.

Given the above considerations, we now give a definition of sub-symmetry for a general differential system. 

\smallskip
Let $q$ smooth functions 
$u=(u^1,u^2,...,u^q)$ of $p$ independent variables $x=(x_1,x_2,...,x_p)$ be defined on some open subset of $\mathbb{R}^p$.
Consider a maximal rank system of $n$ $\ell$th order differential equations 
$\Delta=(\Delta_1,\Delta_2,...,\Delta_n)$ for functions $u$ :
\begin{align}\label{sys}								
\Delta_a(x,u,u_{(1)},u_{(2)},\dots, u_{(l)})=0, \quad a=1,2,...,n.
\end{align}
Here, each $\Delta_a(x,u,u_{(1)},u_{(2)},\dots, u_{(l)} )$,  $a=1,2,...,n$ 
is a smooth function of $x$, $u$, and all partial derivatives of each 
$u^v$, ($v=1,\dots,q$) with respect to the $x_i$ ($i=1,\ldots,p$)
up to the $\ell$th order (i.e. $\Delta_a[u]$ is a differential function, see \cite{Olver}).

\smallskip
Consider an infinitesimal transformation with the canonical infinitesimal operator \eqref{inf_op}
\begin{align}\label{inf_op_J}	
\X = \alpha ^a &\frac{\partial }{\partial u^a} + \sum\limits_i {(D_i 
\alpha ^a)\frac{\partial }{\partial u_i^a }} + \sum\limits_{i \leqslant j} {(D_i } 
D_j \alpha ^a)\frac{\partial }{\partial u_{ij}^a } + \cdots\, = \, \mathrm (D_{J}{\alpha}^a) \partial_{u^a_J}.
\end{align}
where 
$\alpha ^a = \alpha ^a (x,u, u_{(1)}, \dots$), $D_J=D_{j_1 }D_{j_2}\cdot\cdot\cdot D_{j_k}, \quad
D_{j_r}= \frac {d}{dx_{j_r}},\quad r=1,2,. . ., k,$ and the sum is taken over all (unordered) multi-indices $J=(j_1,j_2,...,j_k)$ for $k\ge 0$ and $1\le j_k\le p$, as well as over $1\le a\le q$.

\begin{defn}[Sub-symmetry]
We say that the pair $(\X,\Xi^{ib}\Delta_b)$ is a \textit{sub-symmetry} of system \eqref{sys} if 
\eqal{
\X(\Xi^{ib}\Delta_b)\Big|_{\Delta=0}=0,\qquad i=1,2,\dots,m.
\label{def1}
}
Here, $\Xi^{ib}\Delta_b$ is a sub-system of maximal rank, $\Xi^{ib}=\Xi^{ibJ}[u]D_J$, and $\X$ is an infinitesimal operator \eqref{inf_op_J}. 
\end{defn}
Note that we could, equivalently, require the existence of linear differential operators $\Lambda^{ib}$ with smooth coefficients such that the following relationship holds for any functions $u$:
\begin{align}\label{sub_def}
\text{X}\left(\Xi^{ib}\Delta_{b}\right)=\Lambda^{ib}\Delta_b,\hspace{4 mm} \Lambda^{ib}=\Lambda^{ibJ}[u]D_J,\:\: i=1,2,...,m, \:\: b=1,\ldots,n, \:\: 0<m \le n.
\end{align}
\medskip
Let us briefly discuss the geometric meaning of the definition of sub-symmetries.

\begin{defn}
Let $M$ and $N$ be submanifolds of the jet space $J$.  If $\X$ is a vector field on $J$, we say that $(\X,N)$ is a \textit{sub-symmetry} of $M$ if (i) $M\subset N$  ($M$ is smoothly embedded in $N$),  and (ii) (the prolongation of) $\X$ is tangent to $N$ on points of $M$.  
\end{defn}

In local coordinates, $M$ is the solution set of the differential system 
$\Delta$:  $\Delta_i=0$, $i=1,\dots,n$, and $N\supset M$ is the solution set of some subsystem 
of system $\Delta$ determined by differential operators $\Xi^{ib}$: 
$\Xi^{ib}\Delta_b=0$, $i=1,\dots,m$.

\smallskip

If $M=N$, then vector field $\X$ is tangent to $N$ on points of $N$, 
and therefore $\X$ is a symmetry of $N$.

\smallskip

Note that in this paper we consider classical sub-symmetries, and do not discuss higher sub-symmetries that could be defined on an infinite jet space.

\smallskip

Let us briefly discuss relationships between sub-symmetries and known symmetries. 

\subsection{Relation to symmetries and conditional symmetries}

In the case $m=n$, sub-symmetries turn into regular symmetries of the system (\ref{sys}). Indeed, when $m=n$, and both $\Xi^{ij}$ and $\Lambda^{ij}$ $(i,j =1,\ldots,n)$ are differential functions, the definition (\ref{sub_def}) of a sub-symmetry leads to the classical Lie point symmetry of the system (without loss of generality we could consider $\Xi^{ij}=\delta^{ij}$ here).  The case $m=n$ with $\Xi^{ij}$ as differential functions (for simplicity, $\Xi^{ij}=\delta^{ij}$) and $\Lambda^{ij}$ as differential operators corresponds to tangent transformations or generalized symmetries. Computationally, symmetries and sub-symmetries are quite similar. The determining equations for both look the same in the jet space before substitutions are made.  However, in the case $m<n$, the transformation $\text{X}$ in (\ref{sub_def}) is not a Lie symmetry transformation.  In this sense, the set of sub-symmetries could be considered an extension of the set of Lie symmetries. 

We now compare sub-symmetries with ``non-classical" symmetries \cite{Bluman69}. In the non-classical method, the original system \eqref{sys} is appended by a set of ``invariant surface conditions"
\begin{align}\label{invar_surface}
\X u^a=\phi^a-u^a_i\xi^i =0,  \qquad a=1,2,...,q.
\end{align}
If $E$ is the solution manifold of \eqref{sys}, and $P$ is the solution manifold of \eqref{invar_surface}, then non-classical symmetries were defined as Lie symmetries of the appended system \eqref{sys}- \eqref{invar_surface} on the submanifold $E \cap P$. It is important to point out that the invariance condition $\text{X}\tilde{\Delta}^v=\Lambda^{vj}\tilde{\Delta}^j$ must hold for all equations of the appended system $v=1,2,...,n+q$. The same requirement of invariance of all equations of the original system is taking place for conditional symmetries, when instead of the invariant surface condition, some other differential constraint was added to the system, see e.g. \cite{Fushchych,Olver87}, see also \cite{Levi1989}; this is also the case for partial symmetries \cite{Olver_Vorob,Cicogna01}.  Contrary to this provision, only some combinations of the original equations are required to be invariant with respect to the action of a sub-symmetry $\text{X}$ on {E}, and the other equations may not be (and generally, are not) invariant under $\text{X}$. For example, if our system is $\Delta=(\Delta_1,\Delta_2)$, then a transformation $\text{X}$ that satisfies $\text{X}\Delta_1=A\Delta_1+B\Delta_2$ is a sub-symmetry even if $\text{X}\Delta_2$ is nonzero when $\Delta=0$ ($\Delta_1=\Delta_2=0$).  

\subsection{Properties of sub-symmetries}

We again consider systems of two equations, $\De=(\De_1,\De_2)$.

\bigskip
\noindent
a) {\it{Intersection of sub-symmetry sets}. }
For any differential system, there exist infinitely many sub-systems $\beta^i\Delta_i$, and therefore, infinitely many sub-symmetry vector fields $\X$.  However, they all have a common intersection: the symmetry vector field of the system.  More generally, the intersection  of sub-symmetry vector spaces of any two independent sub-systems $\beta^i\Delta_i$ and $\gamma^i\Delta_i$ (i.e. $\det(\beta^i,\gamma^i)\neq 0$), is precisely the symmetry vector space for the system \eqref{2sys}.  On the other hand, if a differential system has a symmetry vector field, then there exists a sub-symmetry $(\X,\beta^i\Delta_i)$ for every possible choice of $\beta$'s.
\smallskip

\noindent b) {\it{Algebra of sub-symmetries.}}  
The set of all sub-symmetries for a differential equation may not be a Lie algebra.  For example, if $\X_i\Delta_1=\Delta_i$ for two sub-symmetries $(\X_1,\Delta_1)$ and $(\X_2,\Delta_1)$, then $[\X_1,\X_2]\Delta_1=\X_1\Delta_2-\Delta_2$, which, may not be zero on solutions; in this case, the commutator $[\X_1,\X_2]$ is not a sub-symmetry of the system $\Delta_i=0,\: i=1,2$.

\smallskip

\noindent On the other hand, if $\X$ satisfies \eqref{1sym}, or more generally, $\X$ is a sub-system symmetry, then the set of operators $\X$ will have the properties of regular Lie symmetry algebras (restricted to the appropriate sub-system).  This fact allows one to find invariant solutions of the sub-system, conservation laws, a linearization of the sub-system using sub-system symmetries, etc.  

\smallskip

\noindent c) {\it{Sub-symmetry vector field identities.}}

(i) If $(\X,\beta^i\Delta_i)$ is a sub-symmetry and $\text{Y}$ is a symmetry, then $\text{Y}\X(\beta^i\Delta_i)=0$ on solutions 
(i.e. $\text{Y}\X$ is a sub-symmetry operator).

(ii) If $\text{Y}$ is a symmetry of sub-system $\beta^i\Delta_i$, then for every sub-symmetry $(\X,\beta^i\Delta_i),$ 
$\X\text{Y}(\beta^i\Delta_i)=0$ on solutions.

(iii) If $(\X,\beta^i\Delta_i)$ is a sub-symmetry, and $\text{Y}$ is a symmetry of the system \textit{as well as} of the sub-system 
$\beta^i\Delta_i$, then $([\X,\text{Y}],\beta^i\Delta_i)$ is also a sub-symmetry.

(iv) If $\X\Delta_1=a\Delta_2$ and $\text{Y}\Delta_2=0$ on solutions, then $\text{Y}\X\Delta_1=0$ on solutions.  

(v) If $\X\Delta_1=a\Delta_2$ and $\text{Y}\Delta_2=b\Delta_1$, then $\X\text{Y}\Delta_2=\text{Y}\X\Delta_1=0$ on solutions.

\smallskip

In the case of scalar differential equations (DEs) for a function $u(x)$, the set of sub-symmetries forms the Lie symmetry algebra of that equation.  Indeed, suppose our DE is $\Delta[u]=0$, and sub-symmetry $\X$ satisfies $\X(\beta\Delta)=0$ on solutions.  Since $\X(\beta\Delta)=\beta\X\Delta$ on solutions, and $\beta\neq 0$, we conclude that $\X\Delta=0$ on solutions, which shows that $\X$ is actually a symmetry of $\Delta$.  (For systems of DEs the above argument fails).  Note, that for a scalar DE we could introduce appropriate \textit{nonlocal variables} and obtain a corresponding nonlocal system of equations \cite{Bluman_Chev}.  Sub-symmetries of this (nonlocal) system would correspond to nonlocal sub-symmetries of our original equation.

\smallskip

{\it{Finding sub-symmetries.}}

\medskip

The algorithm for finding sub-symmetries of a given system is similar to that for finding regular symmetries of the system.
In order to find sub-symmetries $(\X,\beta^i\Delta_i)$ (infinitesimals $\alpha^i(x_j,u^j,u^j_k,\dots)$ of operator $\X$ \eqref{inf_op} and coefficients $\beta^i$), one must solve invariance condition \eqref{2sub1} for both $\X$ and $\beta^i$, i.e.
\begin{align}\label{2sub3}
\beta^i\,\X\Delta_i\bigl|_{\Delta_1=\Delta_2=0}=0.
\end{align}
To evaluate \eqref{2sub3} on solutions, we would need to solve \eqref{2sys} for the highest derivatives of the functions, (e.g., $u^1_{22}$ and $u^2_{22}$ for a second order system), and substitute these expressions into \eqref{2sub3}.  After this, all other derivatives (e.g. $u^i_{12}$ if $\alpha^i=\alpha^i(x_j,u^j,u^j_k)$) become independent variables, and their coefficients should vanish. Equating the corresponding coefficients to zero generates an overdetermined system bilinear in $\alpha^i$ and $\beta^i$. Compared to finding \textit{symmetries}, this procedure for finding sub-symmetries is somewhat more complicated. Because of the presence of functions $\beta^i$, the determining equations for sub-symmetries are nonlinear in unknown variables.

The knowledge of a particular form of the subsystem coefficients $\beta^i$ (which is  often the case) reduces the problem of finding sub-symmetries to solving an overdetermined linear system \eqref{2sub3} for infinitesimals $\alpha$'s.  A simple example is $\beta^i=\delta^{i1}$. More generally, in cases when we know that system \eqref{2sys} admits conservation laws (at least one), there exist functions $\beta$'s (characteristics of conservation laws, \cite{Olver}) and $F$'s (fluxes) such that
\begin{align}\label{2cl}
\beta^i\Delta_i=D_iF^i.
\end{align}
In this situation, we can make use of a commutator identity for vertical vector fields \eqref{inf_op} (see \cite{Olver}):
\begin{align}\label{comm}
\X D_i=D_i\X.
\end{align}
Therefore, condition \eqref{2sub1} with \eqref{2cl} reduces to the following system for the functions $\alpha^i$:
\begin{align}\label{subcl}
D_i\left(\X F^i\right)\bigl|_{\Delta_1=\Delta_2=0}=0,
\end{align}
which is sometimes easy to evaluate.  In a following section, we discuss the reasons why computing sub-symmetries of a conservation law is important. 

\subsection{Examples}
Consider the one dimensional Euler system for $(u^1,u^2)=(u,v)(x,t)$:
\begin{align}\label{1de}
\begin{split}
\Delta_1&=u_t+uu_x=0,\\
\Delta_2&=v_t+uv_x=0.
\end{split}
\end{align}
Let us look for sub-symmetries of \eqref{1de}.  A natural choice for a sub-system is $\beta^i\Delta_i=\Delta_1$, since it is, itself, a conservation law.  We find several families of sub-symmetries $(\X,\Delta_1)$:
\begin{align}\label{1desub}
\begin{split}
\X_1&=\tau(x,t,u,v)\left[\pd_t+u\pd_x\right],\\
\X_2&=\alpha(x-ut,u,v)\left[t\pd_x+\pd_u\right],\\
\X_3&=\gamma(x-ut,u,v)\,\pd_x,\\
\X_4&=\lambda(x,t,u,v)\,\pd_v,
\end{split}
\end{align}
where operators are presented in non-canonical form, $\tau$ and $\lambda$ are arbitrary functions of $(x,t,u,v)$, and $\alpha$ and $\gamma$ are arbitrary functions of $(x-ut,u,v)$.  

Note that system \eqref{1de} is unusual in that its point symmetry group is only slightly smaller than its collection of sub-symmetries; indeed, the only difference is that for the point symmetries, $\lambda=\lambda(x-ut,u,v)$.  As the next example illustrates, for many systems point symmetry groups are much smaller than their sub-symmetry sets.

Consider the sine-Gordon equation in the following form (nonlocal formulation of $u_{xt}=\sin u$):
\begin{align}\label{sg}
\begin{split}
\Delta_1&=u_x-v=0,\\
\Delta_2&=v_t-\sin u=0.
\end{split}
\end{align}
The point symmetries of this system are simply
\begin{align}\label{sgpt}
\begin{split}
\X_1&=\pd_t,\\
\X_2&=\pd_x,\\
\X_3&=t\pd_t-x\pd_x+v\pd_v.
\end{split}
\end{align}
Let us consider sub-symmetries of $v\Delta_2-\sin u\,\Delta_1=vv_t-u_x\,\sin u$, which is an obvious conservation law of \eqref{sg}.  We find two sub-symmetry vector fields of $(\Y,v\Delta_2-\sin u\,\Delta_1)$:
\begin{align}\label{sgsub}
\begin{split}
\Y_1&=\cot u\,\pd_u-\frac{v}{2}\pd_v,\\
\Y_2&=\frac{1}{v}\left(-\Psi_v\,\pd_x+\Psi_x\,\pd_v\right)+\frac{1}{\sin u}\left(-\Psi_u\,\pd_t+\Psi_t\,\pd_u\right),
\end{split}
\end{align}
where $\Psi=\Psi(x,t,u,v)$ is an arbitrary function.  We can see that the set of sub-symmetries \eqref{sgsub} here is much richer then the set of symmetries \eqref{sgpt}, and includes all symmetries. Indeed,
\begin{align}
\begin{split}
\X_1&=\Y_2,\qquad \qquad\,\,\,\Psi=\cos u,\\
\X_2&=\Y_2,\qquad\qquad\,\,\,\Psi=-v^2/2,\\
\X_3&=\Y_2-\Y_1,\qquad\Psi=t\,\cos u+xv^2/2
\end{split}
\end{align}
Note that sub-symmetry $\Y_1$ is a symmetry of sub-system $v\,\Delta_2-\sin u\,\Delta_1$.  Indeed,  $\Y_1(\beta^i\Delta_i)=-\beta^i\Delta_i$.

\section{Using sub-symmetries to decouple a system}
\label{sec:decoup}
One of the major applications of symmetries is to simplify a system using a symmetry-based mapping.  For example, the hodograph transformation of hydrodynamic-type systems \cite{Tsarev} can be used to linearize these systems \cite{Sheftel2004}, which is related to the fact that hydrodynamic-type systems possess infinite symmetry algebras.  We show that sub-symmetries (in particular, sub-system symmetries) also can be used to simplify the system.  We describe a class of equations for which use of sub-symmetries is considerably more beneficial than regular symmetries.

\subsection{Motivation}
Consider 1D Euler system \eqref{1de}.  We can see that the equation $\Delta_1$ involves only the function $u$ and its derivatives, 
and has no terms depending on function $v$ 
or its derivatives. 
This decoupling of the first equation makes it possible to solve this nonlinear system. We can start with solving the first equation $\Delta_1=0$ for $u(x,t)$, and then substitute it into $\Delta_2$, which now would include only function $v(x,t)$. Both equations of the system \eqref{1de} are clearly, easy to solve.

In this example, we say that the equation $\Delta_1$ is  \textit{decoupled}
since it can be solved for \textit{free variable} $u(x,t)$ independently of the other equations.  More generally,
\begin{defn}\label{def:dec}
A sub-system $\beta^i\Delta_i$ is \textit{decoupled} in \textit{free variable} $u^1$ if there exists a nonzero multiplicative factor $F$ such that
\eqal{
\frac{\pd(F\beta^i\Delta_i)}{\pd u^2_I}=0
}
for all multi-indices $I$ (i.e. the equivalent sub-system $F\beta^i\Delta_i$ does not depend on $u^2$ or its derivatives).
\end{defn}
Since \textit{decoupling} is a desirable property for solving a differential system, it motivates the question of how to predict  when a general, coupled system, such as one for $u^1,u^2=f^a(x_1,\dots,x_m),\: a=1,2, \:m\ge 1$:
\begin{align}\label{arbsys}
\begin{split}
\Delta_1(x,u^a,u^a_i,u^a_{ij},\dots)&=0,\\
\Delta_2(x,u^a,u^a_i,u^a_{ij},\dots)&=0
\end{split}
\end{align}
admits a free variable after applying some transformation.  We will answer this question in two steps: 1. Characterize which equations are decoupled (admit free variables).  2.  Characterize which equations can be transformed into decoupled ones.

Our approach is based on the invariance properties of the system. Could symmetries help to answer the question if a given system can be decoupled? Consider the following nonlinear heat system for $u,v=f(x,t)$:
\begin{align}\label{heat}
\begin{split}
\Delta_1&=u_t-u_{xx}=0,\\
\Delta_2&=v_t-u\,v_{xx}=0.
\end{split}
\end{align}
We see that $\Delta_1$ is decoupled with respect to variable $u$. Are there any indications to this fact in the symmetry group of the system?  It can be shown that the point symmetries of this system are
\begin{align}\label{psymsfail}
\begin{split}
\X_1&=\pdf{t},\qquad\X_2=\pdf{x},\qquad\X_3=\pdf{v},\\
\X_4&=x\pdf{v},\qquad\X_5=v\pdf{v},\qquad\X_6=t\pdf{t}+\frac{1}{2}\,x\pdf{x}.
\end{split}
\end{align}
However, none of these symmetries seem to be able to describe decoupling of $u$. Moreover, if we modify this system by adding inhomogeneous heat sources:
\begin{align}\label{heath}
\begin{split}
\Delta_1&=u_t-u_{xx}+f(x,t)=0,\\
\Delta_2&=v_t-u\,v_{xx}+g(x,t)v+h(x,t)=0,
\end{split}
\end{align}
we can see that $u$ is still a free variable.  However, for general functions $f,g,$ and $h$, this system does not have any point symmetries! 
Therefore, it is not clear if the decoupling of the systems \eqref{heat} or \eqref{heath} is connected to their classical symmetry groups.

\subsection{Sub-system symmetries and mappings}
The reason why symmetries cannot describe the decoupling of a system is related to the fact that symmetries characterize invariance properties of the entire system while decoupling is rather a property of some of its subsystems.  We show that sub-symmetries (particularly, sub-system symmetries) naturally explain the phenomenon of decoupling, and give a straightforward predictive criterion for when a system allows decoupling. 
  
\bigskip
Let us illustrate our idea, and examine the sub-symmetries \eqref{1desub} of the system \eqref{heat}. The sub-symmetry vector field $\X_4=\lambda\pdf{v}$ is a translation in $v$ by an \textit{arbitrary function} $\lambda(x,t,u,v)$ of all variables. However, since equation $\Delta_1$ has no $v$ or its derivatives, sub-symmetry $\X_4$ will not change equation $\Delta_1$: 
\begin{align}
\begin{split}
\X_4\Delta_1&=\X_4(u_t)+\X_4(uu_x)=0.
\end{split}
\end{align}
Note that equality holds for all functions $(u,v),$ including those that are not solutions of $\Delta_2$. We see that equation \eqref{1sym} holds, $\X\Delta_1=a\,\Delta_1$, $a=0$. Therefore, $\X_4$ is a sub-symmetry of the system \eqref{heat}, and symmetry of $\Delta_1$ (sub-system symmetry), and is related to the decoupling of the system \eqref{heat}. For the same reason, $(\X_4,\Delta_1)$ is a sub-system symmetry of the nonlinear heat system \eqref{heath} and is related to its decoupling.

These observations can be generalized. The following theorem gives a criterion for decoupling a system. We use multi-index notation $I=(I_1,\dots,I_n)$ to denote derivative and coordinate subscripts.

\begin{thm}\label{thm:unique}
Consider system \eqref{arbsys} for $u^1,u^2$.  Sub-system $\beta^i\Delta_i$ is decoupled with free variable $u^1$ if and only if system \eqref{arbsys} possesses a sub-symmetry (sub-system symmetry) $(\Xe_\lambda,\beta^i\Delta_i)$ of the form
\begin{align}\label{lam}
\Xe_\lambda=\lambda(x,u)\pdf{u^2},
\end{align}
where $\lambda(x,u)$ is arbitrary. 
\end{thm}
\begin{proof} 
If $\beta^i\Delta_i$ is decoupled with free variable $u^1$, then there exists nonzero factor $F$ such that $F\beta^i\Delta_i$ contains neither $u^2$ nor its derivatives.  Therefore,
\eqal{
\X_\lambda(\beta^i\Delta_i)&=\frac{1}{F}\X_\lambda(F\beta^i\Delta_i)-\frac{\X_\lambda F}{F}\beta^i\Delta_i\\
&=-\frac{\X_\lambda F}{F}\beta^i\Delta_i=0
}
for each $\lambda(x,u)$ and each $(u^1,u^2)$ that solve our sub-system $\beta^i\Delta_i$.  Thus, $(\X_\lambda,\beta^i\Delta_i)$ is a sub-symmetry of the system \eqref{arbsys} (\eqref{2sub} holds), and sub-system symmetry of $\beta^i\Delta_i$ (equation \eqref{1sym} for $\beta^i\Delta_i$ holds, with $a=-\X F/F$).

Conversely, suppose $(\X_\lambda,\beta^i\Delta_i)$ is a sub-system symmetry of $\beta^i\Delta_i$. 
Without loss of generality, we consider $\beta^i\Delta_i=\Delta_1$. Since $\X_\lambda$ is a \textit{point symmetry transformation} of $\Delta_1$, there exists $\Gamma_\lambda$ (see \cite{Olver}) such that
\begin{align}\label{gamlam}
\X_\lambda\Delta_1=\Gamma_\lambda\Delta_1.
\end{align}
Since \eqref{gamlam} holds for all $\lambda's$, we may set $\lambda=\lambda(x)$ to be an arbitrary function of $x=(x_1,\dots,x_m)$. $X_\lambda$ is a canonical infinitesimal operator \eqref{inf_op}, and therefore, terms in the left hand side are of the form $\pdf[\lambda]{x_I}\pdf[\Delta_1]{u^2_I}.$ Then $\Gamma_\lambda$ must be linear in $\lambda$ and $\lambda$ derivatives, or $\Gamma_\lambda=\pdf[\lambda]{x_I}\Gamma^{I}$.  Since $\lambda$ is arbitrary, the coefficients of $\pdf[\lambda]{x_I}$ in both sides should be equal 
\begin{align}\label{u2I}
\pdf[\Delta_1]{u^2_I}=\Gamma^{I}\Delta_1,\qquad \forall I.
\end{align}
Taking the $u^2_J$ derivative and using the fact that the lhs is symmetric with respect to $I$ and $J$:
\begin{align}\label{u2IJ}
\frac{\partial^2\Delta_1}{\partial u^2_I\partial u^2_J}=\left(\pdf[\Gamma^I]{u^2_J}+\Gamma^I\Gamma^J\right)\Delta_1=\frac{\partial^2\Delta_1}{\partial u^2_J\partial u^2_I}=\left(\pdf[\Gamma^J]{u^2_I}+\Gamma^J\Gamma^I\right)\Delta_1.
\end{align}
If $\Delta_1\neq 0$, this implies
\eqal{\label{Gammas}
\pdf[\Gamma^I]{u^2_J}=\pdf[\Gamma^J]{u^2_I}.
}
Since \eqref{2sys} is a nondegenerate system, the solution set of $\Delta_1=0$ defines a lower-dimensional submanifold in the jet space.  Since the $\Gamma$'s are smooth functions, we conclude that equality \eqref{Gammas} holds identically.  Therefore, there exists $\Phi$ such that
\eqal{
\Gamma^I=\pdf[\Phi]{u^2_I},
}
so \eqref{u2I} becomes
\eqal{
\pdf[\Delta_1]{u^2_I}=\pdf[\Phi]{u^2_I}\De_1.
}
We conclude that
\begin{align}\label{delrep}
\Delta_1(x,u^a,u^a_i,\dots)=C(x,u^1,u^1_i,\dots)e^{\Phi(x,u^a,u^a_i,\dots)},
\end{align}
where function $C$ depends only on $x$, $u^1,$ and its derivatives. Therefore, the equivalent sub-system
\eqal{
C(x,u^1,u^1_i,\dots)=e^{-\Phi}\Delta_1=0,
}
does not depend on $u^2$ or its derivatives.  From Definition \ref{def:dec}, we see that $\Delta_1$ is decoupled with free variable $u^1$.
\end{proof}

Note that function $\Gamma_\lambda$ in \eqref{gamlam} exists provided the equation $\Delta_1$ is nondegenerate.

For larger systems with $n$ functions $u^1,\dots,u^n$, there is a natural extension of this result.  A combination $\beta^i\Delta_i$ is decoupled with free variable $u^n$ if and only if it admits $n-1$ sub-system symmetries of the form $\X_i=\lambda^i(x,u)\pdf{u^i}$, $i=1,\dots,n-1$, where each function $\lambda^i(x,u)$ is arbitrary (of all its variables).  There could also be weaker conditions.  For example, in the case of $n=3$ equations and unknowns $u^1,u^2,u^3$, the two-equation sub-system $\beta^{ij}\Delta_j,\:\:i=1,2$ could decouple only 2 variables $(u^1,u^2)$ from the system (i.e. two equations of the system would involve two variables $u^1$ and $u^2$, but not $u^3$). In this case, the system would possess one sub-system $\beta^{ij}\Delta_j$ symmetry $\X_\lambda=\lambda(x,u)\pdf{u^3}$.

\bigskip
We now address the question of which systems have a decoupled sub-system \textit{after} applying some invertible transformation to $(x,u)$. We call such a sub-system \textit{decouplable}.
 
\smallskip
Consider smoothly invertible point transformations $\mc T:(x,u)\to(\ov x,\ov u)$ of the form
\begin{align}\label{pt}
\begin{split}
\ov x_i&=\ov X\,^i(x,u),\qquad i=1,\dots,m,\\
\ov u\,^a&=\ov U\,^a(x,u),\qquad a=1,2.
\end{split}
\end{align}
The associated inverse mapping $\mc T^{-1}:(\ov x,\ov u)\to(x,u)$ is denoted by
\begin{align}\label{pti}
\begin{split}
x_i&=X^i(\ov x,\ov u),\qquad i=1,\dots,m,\\
u\,^a&=U^a(\ov x,\ov u),\qquad a=1,2.
\end{split}
\end{align}

Our approach to systems with decouplable sub-systems builds on work \cite{Bluman_Chev} on symmetries of linearizable systems.

\begin{thm}\label{thm:map}
System \eqref{arbsys} for $(u^1,u^2)=f^a(x_1,\dots,x_m),$ ($a=1,2$) has a decouplable sub-system $\beta^i\Delta_i$ if and only if $\beta^i\Delta_i$ admits a sub-symmetry (sub-system symmetry) $(\Xe_\lambda,\beta^i\Delta_i)$ of the form
\begin{align}\label{vf}
\begin{split}
\Xe_\lambda&=\lambda(x,u)\left[\xi^i(x,u)\pdf{x_i}+\eta^a(x,u)\pdf{u\,^a}\right],
\end{split}
\end{align}
\end{thm}
\noindent where function $\lambda(x,u)$ is arbitrary.
\begin{proof}
If $\beta^i\Delta_i$ is decouplable, then there exists transformation $\mc T:(x,u)\to(\ov x,\ov u)$ given by \eqref{pt} such that the transformed sub-system $(\beta^i\Delta_i\circ\mc T^{-1})(\ov x,\ov u,\dots)$ is decoupled with free variable $\ov u\,^1$.  By Theorem \ref{thm:unique}, $\beta^i\Delta_i\circ\mc T^{-1}$ admits a sub-system symmetry of the form
\begin{align}
\begin{split}
\X_\lambda=\lambda(\ov x,\ov u)\pdf{\ov u\,^2}
=\ov\lambda(x,u)\left[\pdf[X^i]{\ov u\,^2}\pdf{x_i}+\pdf[U^a]{\ov u\,^2}\pdf{u^a}\right],
\end{split}
\end{align}
where function $\lambda$ is arbitrary, $\ov\lambda(x,u)=\lambda(\ov X(x,u),\ov U(x,u))=\lambda\circ\mc T$.  Since $\mc T$ is invertible, function $\ov\lambda$ is also arbitrary, so $\X_\lambda\circ\mc T$ is a sub-system symmetry of the form \eqref{vf}, with $\xi=\pdf[X]{\ov u\,^2}\circ\mc T,$ and $\eta=\pdf[U]{\ov u\,^2}\circ\mc T$.

\medskip
Conversely, suppose $\beta^i\Delta_i$ admits $\X_\lambda$ as a sub-system symmetry.  Since not all $\xi$'s and $\eta$'s are zero, we can use the Frobenius Theorem to find a point transformation $\mc T$ of the form \eqref{pt} that ``straightens out" vector field $\X_1$ ($\lambda=1$), such that, in transformed variables, it becomes a translation in $\ov u\,^2$:
\begin{align}
\X_\lambda\circ\mc T^{-1}=(\lambda\circ\mc T^{-1})\times\,\X_1\circ\mc T^{-1}=\hat\lambda(\ov x,\ov u)\,\pdf{\ov u\,^2},
\end{align}
where $\hat\lambda(\ov x,\ov u)=\lambda\circ\mc T^{-1}=\lambda(X(\ov x,\ov u),U(\ov x,\ov u))$.  Since $\mc T$ is invertible, $\hat \lambda$ is also arbitrary.  But this is a sub-system symmetry of transformed sub-system $\beta^i\Delta_i\circ\mc T^{-1}$ of the form \eqref{lam}.  Therefore, by Theorem \ref{thm:unique}, $\beta^i\Delta_i\circ\mc T^{-1}$ is decoupled with free variable $\ov u\,^1=U^1(x,u)$.
\end{proof}
Since $\lambda(x,u)$ is an arbitrary function of all its $m+2$ variables, and $\lambda(x,u)$ is any representative, then the mappings induced by $\mc T$ and $\mc T^{-1}$ are automorphisms on this set of functions, so 
$\lambda\circ \mc T$ and $\lambda\circ\mc T^{-1}$ are also representatives of this set.
 
It is understood that the ``decoupling" proven in the ``if part" is true up to a nonzero multiplicative factor (see Theorem \ref{thm:unique} and Definition \ref{def:dec}).  We also note that the Frobenius Theorem was used to solve the ``mapping system" (see \cite{Bluman_Chev})
\begin{align}\label{mapsys}
\begin{split}
\X_1\,\ov X\,^i&=\xi^j(x,u)\pdf[\ov X\,^i]{x_j}+\eta^b(x,u)\pdf[\ov X\,^i]{u^b}=0,\qquad i=1,\dots,m,\\
\X_1\,\ov U\,^a&=\xi^j(x,u)\pdf[\ov U\,^a]{x_j}+\eta^b(x,u)\pdf[\ov U\,^a]{u^b}=\delta^{a2},\qquad a=1,2.
\end{split}
\end{align}
Obviously, we require functionally independent solutions such that $\det(X,U)\neq 0$.  Solving these conditions yields an explicit form for the desired point transformation \eqref{pt}.

The extension of this theorem to larger systems for $u^1,\dots,u^n$ is completely analogous, the ``if part" only requiring $n-1$ linearly independent vector fields of the form \eqref{vf}.  Similar considerations hold for decoupling larger sub-systems $\beta^{ij}\Delta_j$ (with more than one equation).

\medskip
Let us also note that the transformation $\mc T$ is not unique, since we can  
replace it with $\mc S\circ\mc T$, where $\mc S$ is any transformation of the following type.  We will reuse equations \eqref{pt}-\eqref{pti} to refer to the component functions of such transformations.
\begin{prop}
If $\Delta_1$ is decoupled with free variable $u^1$, then transformed sub-system $\De_1\circ\mc S^{-1}$ is decoupled with free variable $\ov u\,^1$, where $\mc S$ in \eqref{pt}-\eqref{pti} is any such transformation that satisfies
\eqal{
\frac{\pd X^i}{\pd\ov u\,^2}=\frac{\pd U^1}{\pd\ov u\,^2}=0,\qquad i=1,\dots,m.
}
\end{prop}
\begin{proof}
Since $F\De_1$ does not depend on $u^2$ for some nonzero factor $F$, $(F\De_1)\circ\mc S^{-1}$ does not depend on $U^2$. The latter equation depends only on $X^i$ and $U^1$, which do not contain $\ov u\,^2$.
\end{proof}

\subsection{Determining if a sub-system is decouplable}
\label{sec:alg}
The above results suggest the following methodology to apply to a given differential system, such as \eqref{arbsys}, in order to see if it has a decouplable sub-system.

\medskip
Let $\X_\lambda$ be a sub-system symmetry vector field of the form \eqref{vf}, where $\xi$ and $\eta$ are to be determined, and $\lambda$ is arbitrary.  In general, we let the additional unknowns $\beta^i$ ($\beta^1$ and $\beta^2$) be functions of $(x,u,u_{(1)},\dots,u_{(r-1)})$, where $r$ is the highest order a derivative appears in our system, but simpler choices such as $\beta=\beta(x,u)$ can be also successful.  We solve the \textit{sub-system symmetry} invariance condition
\begin{align}\label{1sub}
\X_\lambda(\beta^i\Delta_i)\bigl|_{\beta^i\Delta_i=0}=0
\end{align}
for the $\xi$'s, $\eta$'s, and $\beta$'s, subject to the requirements that $\lambda$ is an arbitrary function of all variables and not both $\beta$'s are zero.

The procedure here is similar to that for finding ordinary symmetries, with two caveats.  First, to evaluate \eqref{1sub}, we must solve our system for the highest order derivatives, (say $u^2_{22}$ or $u^1_{22}$ for a second order system).  Because the $\beta$'s are present in the invariance condition, we need to consider two cases: $\beta^2=0$, or $\beta^2\neq 0$. Second, the requirement that function $\lambda(x,u)$ be arbitrary stipulates that we equate the coefficients of $\lambda$ and its derivatives to zero, which generates an additional hierarchy to the determining system of equations.  It is this condition that may prevent the existence of sub-system symmetries.

Note that the equivalence of sub-systems under invertible linear transformations 
allows us to make some simplifications to the $\beta$'s in computations.  If, for example, $\beta^2=0$, then we can set $\beta^1=1$ without loss of generality.  On the other hand, if $\beta^2\neq 0$, then we may set $\beta^2=1$ without loss of generality.  After computations are complete, we may multiply both $\beta$'s by some factor $\kappa$ for maximum simplification.

Note also that taking $\lambda=\lambda(x,u)$ arbitrary is often equivalent to taking $\lambda=\lambda(x)$ arbitrary, which is much simpler for computations.  One approach then, would be to first solve \eqref{1sub} using $\lambda=\lambda(x)$ and arrive at functional forms for $\xi,\eta,$ and $\beta$.  Subsequently, one would then return to \eqref{1sub} and now using these results would solve the more general case of $\lambda=\lambda(x,u)$.

\medskip
If a nonzero solution $(\xi,\eta,\beta)$ to \eqref{1sub} can be found, then Theorem \ref{thm:map} guarantees that sub-system $\beta^i\Delta_i$ can be made decoupled with some free variable.  At this point, simply evaluating $\beta^i\Delta_i$ may be enough to find the free variable's functional form.  However, if the mapping $\mc T$ in Theorem \ref{thm:map} is nontrivial, such as hodograph transformation (e.g. $x\leftrightarrow u$), then it becomes necessary to solve mapping system \eqref{mapsys} explicitly.  Solutions $\ov X,\ov U$ will depend on arbitrary functions. The free variable in question is $\ov U\,^1$, and the form of the transformed decoupled sub-system $(\beta^i\Delta_i)\circ\mc T^{-1}$ remains to be found.  Once $\ov X$ and $\ov U$ are obtained, we need to find $(x,u)$ as functions of $(\ov x,\ov u)$, apply the chain rule to evaluate $\pdf[u]{x}$, and substitute this information into the sub-system.  

\subsection{Example: ODEs in polar coordinates}
We consider a class of ordinary differential equations ($m=1$) for $x(t),y(t)$ that arises in dynamical systems \cite{Perko}:
\begin{align}\label{dyn}
\begin{split}
\Delta_1&=x'-x\,F\bigl(x^2+y^2,t\bigr)+y\,G\bigl(x,y,t\bigr)=0,\\
\Delta_2&=y'-y\,F\bigl(x^2+y^2,t\bigr)-x\,G\bigl(x,y,t\bigr)=0,
\end{split}
\end{align}
where $F(\rho,t)$ and $G(x,y,t)$ are some functions.  No sub-system is decoupled in $x$ or $y$. We will show, however, that the system \eqref{dyn} has a decouplable sub-system using its invariance properties.

For general $F$ and $G$, system \eqref{dyn} has no point symmetries.  Let us find its sub-system symmetries, in accordance with Theorem \ref{thm:map}.  For $\beta^i=\beta^i(x,y,t)$, we look for vector fields $\X_\lambda=\lambda[\tau\pd_t+\xi\pd_x+\psi\pd_y]$, where $\tau,\xi,\psi,\lambda=f(x,y,t)$ and $\lambda$ is arbitrary, such that the sub-symmetry condition is verified for every $\lambda$:
\eqal{\label{inv}
\X_\lambda(\beta^1\De_1+\beta^2\De_2)\bigl|_{\beta^1\De_1+\beta^2\De_2=0}=0.
}
There are two cases.  If $\beta^2=0$, then we may set $\beta^1=1$, which means we wish to solve the equation
\eqal{\label{inva}
\X_\lambda\De_1\bigl|_{\De_1=0}=0,\qquad\lambda(x,y,t)\text{ arbitrary}.
}
For general $F$ and $G$, we find no nontrivial solutions.  In the second case when $\beta^2\neq 0$, we may set $\beta^2=1$, which means we wish to solve 
\eqal{\label{invb}
\X_\lambda(\beta^1\De_1+\De_2)\bigl|_{\beta^1\De_1+\De_2=0}\:\:\:=\:0.
}
This equation has the particular solution
\eqal{
\beta^1=x/y,\qquad\tau=0,\qquad \xi=1,\qquad \psi=-x/y.
}

\noindent Therefore, the sub-system
\begin{align}\label{dynsub}
x\De_1+y\De_2=xx'+yy'-(x^2+y^2)F(x^2+y^2,t)=0
\end{align}
has the sub-system symmetry
\begin{align}\label{dynsym}
\X_\lambda=\lambda(x,y,t)\left[-y\pdf{x}+x\pdf{y}\right],
\end{align}
where $\lambda$ is an arbitrary function.  From Theorem \ref{thm:map}, we conclude that $x\De_1+y\De_2$ is a decouplable sub-system.

To illustrate the construction of an explicit free variable, we devise a mapping $\mc T:(x,y,t)\to (\ov x,\ov y,\ov t)=(X(x,y,t),Y(x,y,t),T(x,y,t))$ using the mapping system
\begin{align}
\begin{split}
\X_1\,T(x,y,t)&=-yT_x+xT_y=0,\\
\X_1\,X(x,y,t)&=-yX_x+xX_y=0,\\
\X_1\,Y(x,y,t)&=-yY_x+xY_y=1.
\end{split}
\end{align}
In general, we find that $T,\:X\;=f^a(x^2+y^2,t), \:a=1,2$, while $Y=\arctan(x,y)+g(x^2+y^2,t)$, where $\arctan(x,y)$ denotes the angle in the Cartesian plane.  As the simplest choice, we may set $T=t,\:\:X=\sqrt{x^2+y^2}\equiv r,\:\:Y=\arctan(x,y)\equiv\theta$ (relabeling $\ov x\equiv r,\ov y\equiv\theta$).  

By construction, vector field \eqref{dynsym} in transformed variables becomes
\begin{align}
\X_\lambda\circ\mc T^{-1}=\ov\lambda(r,\theta,t)\pdf{\theta},\qquad\ov\lambda\text{ arbitrary},
\end{align}
while sub-system \eqref{dynsub} becomes
\begin{align}
(x\De_1+y\De_2)\circ\mc T^{-1}=rr'-r^2F(r^2,t)=0,
\end{align}
which is clearly, decoupled with free variable $r(t)$.  

We also have another linearly independent sub-system of \eqref{dyn}, which we present as
\begin{align}
(-y\De_1+x\De_2)\circ\mc T^{-1}=(-yx'+xy'-(x^2+y^2)G)\circ\mc T^{-1}=r^2\theta'-r^2G.
\end{align}
In transformed variables $r,\: \theta$ we thus, have a simplified system with a decoupled sub-system:
\begin{align}
\begin{split}
\ov\De\,^1&=r'-rF(r^2,t)=0,\\
\ov\De\,^2&=\theta'-G(r\cos\theta,r\sin\theta,t)=0,
\end{split}
\end{align}
which recovers the well-known transformation of dynamical system \eqref{dyn} to polar coordinates, which we obtained using sub-symmetries of the system. Note, that the case of $F(r^2,t)=1-r^2$, $G(x,y,t)=1$ corresponds to a solution with a stable limit cycle.

\subsection{Example: Reaction-diffusion systems}
Let us consider a reaction diffusion system for $u(x,t),\;v(x,t)$:
\begin{align}\label{rd}
\begin{split}
\De\,^1&=u_t-D\,u_{xx}-R(u,v)=0,\\
\De\,^2&=v_t-E\,v_{xx}-S(u,v)=0.
\end{split}
\end{align}
Here, constants $D$ and $E$ represent diffusivities, while $R$ and $S$ account for the reaction species interactions.  We require that $R_v$ and $S_u$
be nonzero, so that sub-systems of \eqref{rd} are decoupled neither in $u$ nor $v$.  However, we note that the reaction diffusion system decoupled in $v$
\eqal{
&u_t-Du_{xx}-\al(u-u^2)-\beta v(1-u)=0,\\
&v_t-v_{xx}-(v-v^2)=0,
}
considered in \cite{Holzer}, exhibits anomalous diffusion. We work in $1+1$ dimensions, but the following discussion holds for the multi-dimensional case as well, where $u_{xx}$ is replaced with $u_{xx}+u_{yy}$, for example.

Using the algorithmic procedure outlined in \S\ref{sec:alg}, we classify systems of the form \eqref{rd} that have a decoupled sub-system $\beta^i\Delta_i$, where the $\beta$'s can depend on $(x,t,u,v,u_x,v_x)$.  We invoke Theorem \ref{thm:map} and find its sub-system symmetries.  We look for multipliers $\beta^i$ and vector fields $\X_\lambda=\lambda[\xi\pd_x+\tau\pd_t+\eta\pd_u+\psi\pd_v]$, where $\xi,\tau,\eta,\psi,\lambda=f(x,t,u,v)$ such that invariance condition
\eqref{inv} is satisfied.  

In the case when $\beta^2=0$ and $\beta^1=1$, we solve condition \eqref{inva}.  We find no solutions if $R_v\neq 0$.  In the second case when $\beta^2=1$, we solve condition \eqref{invb}.  If $\xi\neq 0$, we find that $R$ and $S$ must be linear for nontrivial solutions to exist, so we consider the case with $\xi=0$ (nonlinear reaction diffusion).  
As a consequence, we find that $D=E$, such that the two species diffuse at identical rates.  

For this case, we have $\beta^1=-k$ and $\beta^2=1$, where $k$ is a constant.  Here, reaction term $S$ must be of the form
\begin{align}
\begin{split}
S(u,v)=k\,R(u,v)+\sigma(v-k\,u),
\end{split}
\end{align}
where $\sigma(\xi)$ is an arbitrary function.  The decoupled sub-system 
\begin{align}
-k\De_1+\De_2=(v-k\,u)_t-D(v-k\,u)_{xx}-\sigma(v-k\,u)=0
\end{align}
admits the sub-system symmetry
\begin{align}
\X_\lambda=\lambda(x,t,u,v)\left[\,\pdf{u}+k\,\pdf{v}\right],\qquad \lambda\text{ arbitrary}.
\end{align}
Therefore, the free variable is $\ov u=v-k\,u$.  We complete the transformation $\mc T:(u,v)\to(\ov u,\ov v)$ by letting $\ov v=u$.  Setting $\ov\De\,^1=(-k\De_1+\De_2)\circ\mc T^{-1}$ and $\ov\De\,^2=\De\,^1\circ\mc T^{-1}$, we rewrite a linear tranformation of system \eqref{rd} in transformed variables as follows:
\begin{align}
\begin{split}
\ov\De\,^1&=\ov u_t-D\ov u_{xx}-\sigma(\ov u)=0,\\
\ov\De\,^2&=\ov v_t-D\ov v_{xx}-\ov R(\ov u,\ov v)=0,
\end{split}
\end{align}
where $\ov R(\ov u,\ov v)=R(\ov v,\ov u+k\,\ov v)$ is a transformed reaction term.  The first equation is the 1D Kolmogorov-Petrovsky-Piskunov (KPP) equation \cite{KPP}, which is a decoupled sub-system that can be solved independently.  For example, the case $\sigma(\ov u)=r\,\ov u(1-\ov u)$ gives Fisher's equation \cite{Ablowitz}, for which traveling wave solutions are well known.  On the other hand, the second equation is a driven KPP equation with reaction term that depends on free variable $\ov u$, similar to that studied in \cite{Holzer}.  Note that this system, despite having a decoupled sub-system, is still general enough to admit interesting behavior, such as the Turing instability \cite{Turing}.

\medskip

\begin{rem}
We showed that sub-symmetries completely characterize systems with decouplable 
sub-systems.  In fact, sub-symmetries can play an even more important role for some classes of decoupled systems. Consider, for example,  
decoupled linear sub-systems.  1D Euler \eqref{1de}, and nonlinear heat systems \eqref{heat}, and \eqref{heath} are all examples of systems with a decoupled linear equation.  Indeed, for these systems $\Delta_2$ is linear in $v$, and no other equation depends on $v$ (which is another kind of decoupling). As earlier, for these systems we could first solve the $\Delta_1$ equation for function $u(x,t)$, and after substituting this $u$ into equation $\Delta_2$, the obtained equation can be considered linear, with $u$ being a known variable coefficient. Therefore, for systems with decoupled linear sub-systems, nonlinear problems can be reduced to essentially, linear ones.
\end{rem}

\section{Vector field deformations of conservation laws}
\label{sec:deform}
The connection between conservation laws of a differential system and its symmetry properties has a long history, starting with the seminal paper by Noether \cite{Noether} for Lagrangian (variational) systems.  In our next paper \cite{RSII2016}, we study the connection between sub-symmetries and conservation laws for non-Lagrangian systems in detail, and show that in terms of the Noether identity, this connection is quite natural. Here, we discuss some practical aspects of this connection, and consider the action of sub-symmetry transformations on known conservation laws of a differential system. 

\subsection{Symmetry deformations}
It is well known (see e.g. \cite{Olver,Bluman_Chev}) that the application of a symmetry transformation to a conservation law gives back a conservation law, which may or may not be new.  Let us again consider system \eqref{2sys}.  Suppose it has a conservation law of the form
\begin{align}\label{2con}
D_1A^1+D_2A^2\bigl|_{\Delta=0}=0,
\end{align}
where the $A$'s are some functions of $(x_i,u^i,u^i_j,\dots)$.  Equivalently, there exist $\Gamma$'s such that  
\begin{align}
D_iA^i=\Gamma^{j}\Delta_j+\Gamma^{jk}D_k\Delta_j+\dots.
\end{align}
If $\X$ is a symmetry of \eqref{2sys}, then applying $\X$ to this equality, using Leibniz rules, and employing commutation formula \eqref{comm} give
\begin{align}
D_i(\X A^i)=(\X\Gamma^j)\Delta_j+\Gamma^j(\X\Delta_j)+(\X\Gamma^{jk})D_k\Delta_j+\Gamma^{jk}D_k(\X\Delta_j)+\dots.
\end{align}
Since $\X$ is a symmetry, the right hand side vanishes on solutions:
\begin{align}
D_i(\X A^i)\bigl|_{\Delta=0}=0,
\end{align}
which is a conservation law.  We say that vector field $\X$ has \textit{deformed} conservation law $D_iA^i$ into conservation law $D_i(\X A^i)$.  Depending on $\X$ and the $\Gamma$'s, this conservation law may be trivial (no conserved integrals), the same as \eqref{2con},  or entirely different.  

Regarding the inverse problem; can we deform two conservation laws $D_iA^i$ and $D_iP^i$ into each other using a symmetry transformation $\X$, such that $P^i=\X A^i$?  A general answer to this question is negative, since the set of local symmetry transformations is usually rather restricted.  Even for variational problems, with one-to-one correspondence between variational symmetries and local conservation laws (\cite{Olver}), a corresponding symmetry deformation often does not exist.

Consider for example, the Hopf (Euler, Burgers) equation for $u(x,t)$ 
\begin{align}\label{hopf}
\Delta=u_t+\left(\frac{1}{2}u^2\right)_x=0.
\end{align}
This equation is, itself, a conservation law $D_iA^i=0$.  It also has an infinite set of conservation laws $D_iP^i=0$ \cite{Rosenhaus07}:
\begin{align}\label{infhopf}
f''(u)\Delta=(f'(u))_t+(uf'(u)-f(u))_x=0,
\end{align}
where $f(u)$ is an arbitrary function ($f''\neq 0$).  It is easy to show that there is no symmetry $\X=\alpha\pd_u$ of \eqref{hopf} such that $P^i=\X A^i$. Indeed, in this case
\begin{align}\label{comhopf}
\begin{split}
P^t =X A^t = Xu =\alpha& =f'(u),\\
P^x =X A^x = X(u^2/2) =u\alpha&=uf'(u)-f(u),
\end{split}
\end{align}
which is clearly, not possible to satisfy for $f(u)\ne 0$. Adding a gauge-type transformation $P^t\to P^t-D_xR,\:\: \,P^x\to P^x+D_tR$ with an arbitrary function $R$ does not change the conclusion:
\begin{align}\label{comphopf}
\begin{split}
\alpha&=f'(u)-D_xR,\\
u\alpha&=uf'(u)-f(u)+D_tR.
\end{split}
\end{align}
Multiplying the first equation by $u$ and subtracting from the second gives
\begin{align}
0=uD_xR+D_tR-f(u).
\end{align}
If $R=R(x,t,u)$, then the first term involves $u_x$ and the second $u_t$, which do not balance each other, so $R=R(x,t)$.  Similar logic applies to the case of $R=R(x,t,u,u_x,u_t,\dots)$, so we conclude that $R=R(x,t)$.  Therefore,
\eqal{
R_t+uR_x-f(u)=0.
}
We had assumed $f''\neq 0$ for a nontrivial conservation law, so all terms must be zero for consistency, a contradiction.  We conclude that such a deformation is impossible for \textit{any} vector field $\X_f$. Therefore, \eqref{hopf} cannot be deformed to \textit{itself}, since that transformation would correspond to $f''=1$.

Another example is the sine-Gordon equation.  Sine-Gordon system \eqref{sg} has only one conservation law with low order fluxes.  
\begin{align}\label{sgcl}
v\,\Delta_2-\sin u\,\Delta_1=\left(\frac{1}{2}\,v^2\right)_t+\left(\cos u\right)_x=0
\end{align}
However, it can be shown that point symmetries $\X$ of this system \eqref{sgpt} deform non-trivial conservation law \eqref{sgcl} to a trivial one. For example, using a vertical form of $\X_1$, $\X_1=-u_t\pd_u-v_t\pd_v$, we obtain 
\begin{align}
\X_1\left(v\,\Delta_2-\sin u\,\Delta_1\right)=-\left(v\,v_t\right)_t+\left(u_t\,\sin u\right)_x=-D_t\left(v\,\Delta_2-\sin u\,\Delta_1\right).
\end{align}
The flux of this conservation law vanishes on solutions, and we get a trivial conservation law.  We see that such symmetry deformation $\X$ does not exist.

\subsection{Sub-symmetry deformations and the inverse problem}
Compared to symmetry transformations that are too restrictive to 
be able to deform conservation laws, sub-symmetries are considerably better equipped to transform conservation laws into each other. We will show that there are many cases where {\it{sub-symmetry deformations}} of conservation laws exist and are even one-to-one.  In these cases, sub-symmetries provide an effective approach for generating new conservation laws of a differential system.

To demonstrate that conservation law deformations are possible for sub-symmetries, we recall \eqref{subcl}, which expresses the sub-symmetry invariance of sub-system \eqref{2cl} in the form of a conservation law $\beta^i\Delta_i=D_iF^i$.  But condition \eqref{subcl} is precisely a conservation law with deformed fluxes $\X F^i$.  Therefore, we see that similarly to symmetries, sub-symmetries deform conservation laws to conservation laws, and commutation identity \eqref{comm} makes this deformation possible. Thus, we proved the following theorem:
\begin{thm}
Let $(\Xe,\beta^i\Delta_i)$ be a sub-symmetry.  If $\beta^i\Delta_i=D_iF^i$ is a conservation law, then $\Xe(\beta^i\Delta_i)=D_i(\Xe\,F^i)$ is also a conservation law.
\end{thm}

\begin{rem}
To check if $D_i(\X\,F^i)$ is a nontrivial conservation law, we compute its characteristic.  Since $(\X,\beta^i\De_i)$ is a sub-symmetry, there exist $\Gamma$'s such that
\eqal{
\X(\beta^i\Delta_i)&=\Gamma^{iI}D_I\De_i\\
&=\bigl[(-D)_I\Gamma^{iI}\bigr]\De_i+D_j\Phi^j_{I}\bigl[\Gamma^{iI},\De_i\bigr],
}
where $(-D)_I=(-D)_{I_1}\dots(-D)_{I_k}$ is the adjoint operator, and $\Phi^j_{I}$ is a bilinear expression that vanishes when $\De=0$ (i.e. a trivial flux); see \cite{RS2016}.  Therefore,
\eqal{
D_i(\X\,F^i)=\bigl[(-D)_I\Gamma^{iI}\bigr]\De_i+D_j\Phi^j_{I}\bigl[\Gamma^{iI},\De_i\bigr].
}
This conservation law is nontrivial if its characteristic $(-D)_I\Gamma^{iI}$ is nonzero.
\end{rem}

Note that a sub-symmetry $\X$, in general, will only yield a conservation law deformation when applied to its associated sub-system $\beta^i\Delta_i=D_iF^i$.  If we apply $\X$ to another conservation law, $D_iG^i=0$, the resulting divergence expression $D_i\X G^i$ may not be zero on solutions, hence it may not be a conservation law.  This is because $\X$ is only required to leave its corresponding sub-system invariant, not the entire system.  Therefore, unlike for symmetries, the product of two sub-symmetries $\X_1\X_2$ will not produce new conservation laws unless both sub-symmetries correspond to the same sub-system.

\bigskip
Let us look again at the inverse problem of deforming \eqref{sgcl} to itself, this time not by a symmetry transformation but by some vector field $\X=\alpha\pd_u+\beta\pd_v$.  If such a vector field exists, it should satisfy the following system:
\begin{align}\label{sgsys}
\begin{split}
v\,\beta&=\frac{1}{2}\,v^2-D_xR,\\
-\sin u\,\,\alpha&=\cos u+D_tR,
\end{split}
\end{align}
where $R$ is an arbitrary function that we may choose so as to make a solution exist.  We can see that the choice $\beta=v/2,\,\,\alpha=-\cot u, \;\; R=0$ satisfies \eqref{sgsys}, and a vector field $\X$ deforms \eqref{sgcl} to itself.  Since \eqref{sgcl} vanishes on solutions, we also conclude that $(\X,v\Delta_2-\sin u\,\Delta_1)$ is a sub-symmetry (cf. \eqref{sgsub}). Thus, a sub-symmetry $\X=-\cot u \pd_u+\frac{1}{2}v \pd_v$ deforms the conservation law \eqref{sgcl} to itself.  There is nothing special about this example, and the deformation of \eqref{sgcl} to itself could be given a more general consideration. 

Let $D_tP+D_xQ=0$ be \textit{any} conservation law of \eqref{sg}.  Similarly to the example, we conclude that the transformation
\begin{align}
\X=-\frac{1}{\sin u}\,Q\pd_u+\frac{1}{v}\,P\pd_v
\end{align}
deforms \eqref{sgcl} to $D_tP+D_xQ=0$, which implies that $(\X,v\Delta_2-\sin u\,\Delta_1)$ is a sub-symmetry.  

Thus, the set of sub-symmetries of sine-Gordon equation maps one of its conservation law onto the set of conservation laws. This correspondence between sub-symmetries and conservation laws is one-to-one if our sub-symmetry transformations $\X$ are defined up to vector (gauge) fields of the form
\begin{align}
\X_R=-\frac{1}{\sin u}\,D_tR\pd_u-\frac{1}{v}\,D_xR\pd_v,
\end{align}
where $R$ is an arbitrary function.  (In algebraic language, the set of vector fields $\X_R$ is the kernel of the vector space homomorphism: sub-symmetries $\to$ conservation laws, and the associated quotient space of sub-symmetries is isomorphic to the set of conservation laws.)

\medskip
The reason why 
we were able to establish a one-to-one correspondence between sub-symmetries and conservation laws (up to ``gauge" transformations) is that the conservation law \eqref{sgcl} is ``flexible enough" so as to permit solving the deformation system \eqref{sgsys} for $\X$.  The same construction works in more general situations.

\begin{thm}\label{thm:ip}
Suppose that two-dimensional system \eqref{2sys} possesses a nontrivial conservation law of the form
\eqal{\label{clA}
D_1A^1(x,u)+D_2A^2(x,u)=0,
}
such that 
\begin{align}\label{ndeg}
|\pd A|:=\frac{\pd(A^1,A^2)}{\pd(u^1,u^2)}\neq 0.
\end{align}
Then for each conservation law $D_iP^i=0$ of system \eqref{2sys}, there exists a sub-symmetry of \eqref{clA} of the form
\eqal{
\Xe=\frac{1}{|\pd A|}\left[(A^2_{u^2}P^1-A^1_{u^2}P^2)\pd_{u^1}+(A^1_{u^1}P^2-A^2_{u^1}P^1)\pd_{u^2}\right]
}
that deforms $D_iA^i$ to $D_iP^i$.  This sub-symmetry is unique up to gauge sub-symmetries $(\Xe_R,D_iA^i)$ of the form 
\eqal{
\Xe_R=\frac{1}{|\pd A|}\left[-(A^2_{u^2}D_2R+A^1_{u^2}D_1R)\pd_{u^1}+(A^1_{u^1}D_1R+A^2_{u^1}D_2R)\pd_{u^2}\right].
}
\end{thm}

\noindent
\begin{rem}
The uniqueness of the subsymmetry here holds only with respect to sub-system \eqref{clA}.  There could be two different sub-symmetries $(\X_1,D_iA^i)$ and $(\X_2,D_iB^i)$ of two different sub-systems which both deform their respective sub-systems to $D_iP^i$.  But since they are conservation laws, these sub-systems can be deformed to each other; thus, we can identify all conservation laws that satisfy condition \eqref{ndeg} as an \textit{equivalence class} with respect to sub-symmetry deformation.  Conservation laws of type \eqref{clA} are sometimes said to yield ``hydrodynamic integrals" \cite{Tsarev}.  
\end{rem}

\bigskip
There is an analogous result for differential systems of $n$ equations for $q$ functions $u^a(x_1,\dots,x_p)$.  A conservation law $D_iA^i(x,u)=0$ can be deformed to any conservation law $D_iP^i=0$ by a (unique up to gauge) sub-symmetry if the rank of Jacobian matrix $\pd A/\pd u$ is equal to $p$.  Theorem \ref{thm:ip} is formulated for $n=p=q=2$.

\bigskip

Consider now the nonlinear telegraph system studied in \cite{Bluman05} for $u,v=f(x,t)$:
\begin{align}\label{tele}
\begin{split}
\Delta_1&=u_t-v_x=0,\\
\Delta_2&=v_t-F(u)u_x-G(u),
\end{split}
\end{align}
which is a nonlocal formulation of a nonlinear telegraph-type equation:
\begin{align}
u_{tt}-F(u)u_{xx}-F'(u)u_x^2-G'(u)u_x=0,
\end{align}
where $F$ and $G$ are arbitrary functions, see \cite{Bluman05} and references therein.

Given any conservation law of \eqref{tele}, 
\begin{align}\label{tele_CL}
D_t P+ D_x Q \,=\,0,  
\end{align}
we seek its possible deformation by sub-symmetry vector field 
\begin{align}\label{tele_SYM}
\X=\alpha\pd_u+\beta\pd_v
\end{align}
from sub-system $\Delta_1$.

For system \eqref{tele}, $\Delta_1$ is a conservation law \eqref{clA} with $A^1=u$, $A^2=-v$, $x_1=t$, $x_2=x$. 
We have $\partial(A^1,A^2)/\partial(u,v)=-1$, 
so 
condition \eqref{ndeg} is satisfied, and we can apply Theorem \ref{thm:ip}, 
which means that we can solve the deformation equations
\begin{align}\label{def}
\begin{split}
\alpha&=P-D_xR,\\
-\beta&=Q+D_tR.
\end{split}
\end{align}
Thus, $\X=P\pd_u-Q\pd_v$ is the unique sub-symmetry vector field corresponding to $P_t+Q_x=0$, up to a gauge transformation of the form $\X_R=-D_xR\pd_u-D_tR\pd_v$.  

There are many systems with nonlocal formulations that admit a sub-symmetry/conservation law correspondence.  A clear generalization of \eqref{tele} is to replace $\Delta_2$ with $v_t-F(x,t,u,u_x,u_{xx},\dots)$.  In general, the key condition to consider for a given nonlocal formulation is \eqref{ndeg}, and $\Delta_1$ in \eqref{tele} is the prototypical example 
that satisfies this requirement.

We can define a nonlocal formulation for \textit{any} PDE system in a rather obvious way as follows.  Consider a scalar PDE for $u=u(x_1,x_2)$, and introduce new dependent variables: $p^1=u_1$ and $p^2=u_2$.  For compatibility, we need to add $p^1_2-p^2_1=0$ to our nonlocal formulation.  But this equation is a conservation law, and $\partial(-p^2,p^1)/\partial(p^1,p^2)=1$, so Theorem \ref{thm:ip} applies,  with $\X=\alpha\pd_{p^1}+\beta\pd_{p^2}$.  If we have more than two $x$'s (say, $m$), then we can make more $p$'s, and more compatibility relations of the form $p^i_j-p^j_i$ will hold ($m(m-1)/2$); adding together $m-1$ of them yields a conservation law of the desired type.  

Another approach would be to attempt to generalize condition \eqref{ndeg} to the situation where fluxes $A^i$ depend on derivatives $u_1,u_2,u_{11},\dots$, which is often the case for many equations in mathematical physics.  The deformation equations from the conservation law 
$D_1A^1+D_2A^2=0$ to the conservation law $P_1+Q_2=0$ in the case $A^i=A^i(x,u)$ have the following form:
\begin{align}
\begin{split}
\frac{\partial A^1}{\partial u^1}\,\alpha+\frac{\partial A^1}{\partial u^2}\,\beta&=P-D_2R,\\
\frac{\partial A^2}{\partial u^1}\,\alpha+\frac{\partial A^2}{\partial u^2}\,\beta&=Q+D_1R.
\end{split}
\end{align}
In the case of $A^i$ depending on derivatives, we use Fr\'echet derivatives to obtain:
\begin{align}\label{deffre}
\begin{split}
\mbb D_{u^1}(A^1)\,\alpha+\mbb D_{u^2}(A^1)\,\beta&=P-D_2R,\\
\mbb D_{u^1}(A^2)\,\alpha+\mbb D_{u^2}(A^2)\,\beta&=Q+D_1R,
\end{split}
\end{align}
where \cite{Olver}
\begin{align}
\mbb D_{u}(A)=\frac{\partial A}{\partial u}+\frac{\partial A}{\partial u_i}\,D_i+\dots+\frac{\partial A}{\partial u_I}D_I+\dots
\end{align}
is the Fr\'echet derivative of $A$ with respect to $u$, and $I$ denotes a multi-index.

Note that a solution $(\alpha,\beta)$ of system \eqref{deffre} may not exist in some cases, e.g. when $P$ and/or $Q$ depend nonlinearly on their highest order derivatives. 

\medskip
Another scenario of violating condition \eqref{ndeg} is the vanishing of the matrix determinant. In this case, the conserved vector $(A^1,A^2)$ depends, effectively, on only one combination of $(u^1,u^2)$, so we can make a change of variables $(u^1,u^2)\to (w,z)$ to this effect, such that $A^i=A^i(w)$ only (suppressing $x$ dependence).  Thus, our conservation law 
\eqref{clA} looks like $D_1A^1(w)+D_2A^2(w)=0$ (we could also consider another conservation law with ``z" dependence and take a linear combination of these two conservation laws, which would yield a full rank matrix). This conservation law is analogous to the situation with \eqref{hopf}, in that not enough dependent variables are present.  If we proceed the same way, starting from \eqref{comhopf} with 
$A^1(w)$ and $A^2(w)$ in place of $u$ and $u^2/2$, we arrive at the following condition:
\eqal{\label{condition}
A^2_w(P^1-D_2R)-A^1_w(P^2+D_1R)=0.
}
Let us assume that $P^i=P^i(x,u)$, for example.  Then $R=R(x)$.  If $A^2_w\neq 0$, we would get
\eqal{
P^1_w=\pd_w(A^1_w/A^2_w)(P^2+D_1R)+A^1_wP^2_w/A^2_w.
}
If $A^1_w/A^2_w=:B$ is constant in $w$, then condition \eqref{condition} yields
\eqal{
P^1_w=BP^2_w.
}
If $B$ is not constant in $w$, then
\eqal{
\pd_w(P^1_w/B_w)=P^2_w+\pd_w(BP^2_w/B_w).
}
In either case, we see that the $P^i$'s cannot correspond to arbitrary conservation laws. Unless the $P$'s are of specific forms depending on the $A$'s, sub-symmetry deformations here will not be possible.

\subsection{Sub-symmetries of the nonlinear telegraph system}
In this section we analyze the sub-symmetries of the nonlinear telegraph system \eqref{tele} and show that they generate all local conservation laws of this system through deformation. Bluman and Temuerchaolu in \cite{Bluman05} and \cite{Blu05} have done a thorough analysis of a correspondence between classical Lie point symmetries and lower conservation laws for the nonlinear telegraph system \eqref{tele}. They generated several tables comparing point symmetries admitted by the system with their conservation laws of the form $\xi^1(t,x,u,v)\Delta_1+\xi^2(t,x,u,v)\Delta_2$ obtained with the ``direct method" \cite{Bluman_Chev}.  They demonstrated that for some functions $F(u)$ and $G(u)$, the set of conservation laws of this system is larger than the set of its symmetries.  For instance, in the case $F(u)=G(u)=\tan u$ the system has two point symmetries but four conservation laws (see Table 1 in \cite{Bluman05} and Table 8 in \cite{Blu05}).  They concluded that, for these cases, symmetries of the system \eqref{tele} do not generate all of its conservation laws.

However, it follows from Theorem \ref{thm:ip} and the deformation equations \eqref{def} that any conservation law \eqref{tele_CL}: $D_tP+D_xQ=0$ of the system \eqref{tele} could be generated by the sub-symmetry $(\X,\:\Delta_1)$, where $\X=P\pd_u-Q\pd_v.$ Indeed, 
\begin{align}
\X\Delta_1=\X(u_t-v_x)=D_tP-D_x(-Q).
\end{align}

As an explicit example, consider 
the following four conservation laws of the system \eqref{tele} for $F(u)=G(u)=\tan u $ (\cite{Blu05}):
\begin{align*}
&P_{1\pm}=\frac{1}{\sqrt{2}}(\cos u)^{\pm 1}e^{u\pm 2x+\sqrt{2}(v\mp t)},\\
&Q_{1\pm}=\pm\frac{1}{2}\left(\tan u\mp 1\right)(\cos u)^{\pm 1}e^{u\pm 2x+\sqrt{2}(v\mp t)},\\
&P_{2\pm}=\frac{1}{\sqrt{2}}(\cos u)^{\pm 1}e^{u\pm 2x-\sqrt{2}(v\mp t)},\\
&Q_{2\pm}=\pm\frac{1}{2}\left(\tan u\mp 1\right)(\cos u)^{\pm 1}e^{u\pm 2x-\sqrt{2}(v\mp t)}.
\end{align*}
Half of these conservation laws can be accounted for by reflection symmetry $(t,v)\to(-t,-v)$, which still leaves two conservation laws that cannot be accounted for by point symmetries $\pd_t$ and $\pd_x$.  On the other hand, it is straightforward to obtain the sub-symmetries that generate all these conservation laws:
\begin{align*}
&\X_{1\pm}=\frac{1}{2}(\cos u)^{\pm 1}e^{u\pm 2x+\sqrt{2}(v\mp t)}\left[\partial_u+\sqrt{2}(1\mp\tan u)\partial_v\right],\\
&\X_{2\pm}=\frac{1}{2}(\cos u)^{\pm 1}e^{u\pm 2x-\sqrt{2}(v\mp t)}\left[\partial_u+\sqrt{2}(1\mp\tan u)\partial_v\right].
\end{align*}

As another example, consider the case where $F(u)$ is arbitrary and $G(u)=u$ (Table 1 in \cite{Blu05}).  In this case, we have two conservation laws.  One maps to the other by a point symmetry ($\pd_t$), but there is still one conservation law left over unaccounted for by symmetries:
\begin{align*}
&P=(x-t^2/2)u+tv,\\
&Q=(t^2/2-x)v-t\int^uF(s)\di s.
\end{align*}
The following sub-symmetry generates this conservation law:
\begin{align*}
\X=[(x-t^2/2)u+tv]\pd_u+[(x-t^2/2)v+t\int^uF(s)\di s]\pd_v.
\end{align*}

In the case $F(u)=e^u+\alpha^2$ and $G(u)=e^u$, the system \eqref{tele} has four conservation laws (Table 6 in \cite{Blu05}), but only three can be obtained from symmetry transformations $\pd_t,$ and $(t,v)\to(-t,-v)$.  The remaining conservation law is:
\begin{align*}
&P=e^{x+\alpha t}[e^u+A^2/2],\\
&Q=-e^{x+\alpha t}[(A-\alpha)e^u+\alpha A^2/2],
\end{align*}
where $A=2\alpha(x+\alpha t)+(v+\alpha u-\alpha)$.  The sub-symmetry generating this conservation law can be easily found:
\begin{align*}
\X=e^{x+\alpha t}[e^u+A^2/2]\partial_u+e^{x+\alpha t}[(A-\alpha)e^u+\alpha A^2/2]\partial_v.
\end{align*}

The case of $F(u)$ arbitrary and $G(u)=1/u$ (Table 1 in \cite{Blu05}) corresponds to wave propagation in a hyper-elastic homogeneous rod with exponentially varying cross section \cite{jeffrey}.  In this case, we again have two conservation laws, but although one maps via a point symmetry $\pd_v$ to the other, the first is still unaccounted for using symmetries:
\begin{align*}
&P=(x+v^2/2)u+\int^u\int^szF(z)\di z\di s,\\
&Q=-xv-v^3/6-v\int^usF(s)\di s.
\end{align*}
The sub-symmetry that generates this conservation law is as follows:
\begin{align*}
&\X=[(x+v^2/2)u+\int^u\int^szF(z)\di z\di s]\pd_u+[xv+v^3/6+v\int^usF(s)\di s]\partial_v.
\end{align*}

\medskip
Thus, all conservation laws of the system \eqref{tele} unaccounted for by its point symmetries can be obtained through conservation law deformation by the sub-symmetries of the system.

\section{Conclusions}
We introduced the notion of a sub-symmetry of a differential system and  discussed the geometrical meaning and main properties of sub-symmetries. We demonstrated some advantages of sub-symmetries compared to regular symmetries for two aspects. We showed how sub-symmetries can be used in decoupling (and eventually, solving) the system. We also discussed the role of sub-symmetries in conservation law deformation, and the possibility to generate new conservation laws.

\end{document}